\documentclass{llncs}
\pdfoutput=1
\usepackage{amsmath}
\usepackage{amssymb}
\usepackage{latexsym}
\usepackage{algorithm}
\usepackage{algpseudocode}
\usepackage{subfigure}
\usepackage{float}
\usepackage{graphicx}
\newcommand{\bigo}{\mathcal{O}}
\renewcommand{\mod}{\text{ mod }}
\newcommand{\E}{\mbox{\rm \bf{E}}}
\renewcommand{\Pr}{\mbox{\rm \bf{Pr}}}
\newcommand{\Var}{\mbox{\rm \bf{Var}}}
\newcommand{\Xac}{X_{(a,c)}}
\newcommand{\mcA}{\mathcal{A}}
\newcommand{\mcB}{\mathcal{B}}
\newcommand{\mcC}{\mathcal{C}}

\begin{document}
\mainmatter              
\title{Better size estimation for sparse matrix products\thanks{This work was supported by the Danish National Research Foundation, as part of the project ``Scalable Query Evaluation in Relational Database Systems''. A shorter version of this paper has been accepted for presentation at the 14th Intl.~Workshop on Randomization and Computation - RANDOM 2010.}}
\author{Rasmus Resen Amossen \and Andrea Campagna \and Rasmus Pagh}
\institute{IT University of Copenhagen, DK-2300 Copenhagen S, Denmark\\
\email{\{resen,acam,pagh\}@itu.dk}}
\maketitle              

\begin{abstract}
We consider the problem of doing fast and reliable estimation of the number $z$ of non-zero entries in a sparse boolean matrix product.
This problem has applications in databases and computer algebra.

Let $n$ denote the total number of non-zero entries in the input matrices.
We show how to compute a $1\pm \varepsilon$ approximation of $z$ (with small probability of error) in expected time $\bigo(n)$ for any $\varepsilon > 4/\sqrt[4]{z}$.
The previously best estimation algorithm, due to Cohen (JCSS~1997), uses time $\bigo(n/\varepsilon^2)$.
We also present a variant using $\bigo(\text{sort}(n))$ I/Os in expectation in the cache-oblivious model.

In contrast to these results, the currently best algorithms for computing a sparse boolean matrix product use time $\omega(n^{4/3})$ (resp.~$\omega(n^{4/3}/B)$ I/Os), even if the result matrix has only $z=\bigo(n)$ nonzero entries.

Our algorithm combines the size estimation technique of Bar-Yossef et al.~(RANDOM 2002) with a particular class of pairwise independent hash functions that allows the sketch of a set of the form $\mcA\times \mcC$ to be computed in expected time $\bigo(|\mcA|+|\mcC|)$ and $\bigo(\text{sort}(|\mcA|+|\mcC|))$ I/Os.

We then describe how sampling can be used to maintain (independent) sketches of matrices that allow estimation to be performed in time $o(n)$ if $z$ is sufficiently large. 
This gives a simpler alternative to the sketching technique of Ganguly et al.~(PODS 2005), and matches a space lower bound shown in that paper.

Finally, we present experiments on real-world data sets that show the accuracy of both our methods to be significantly better than the worst-case analysis predicts.
\end{abstract}



\section{Introduction}

In this paper we will consider a $d\times d$ boolean matrix as the subset of $[d]\times [d]$ corresponding to the nonzero entries.
The product of two matrices $R_1$ and $R_2$ contains $(i,k)$ if and only if there exists $j$ such that $(i,j)\in R_1$ and $(j,k)\in R_2$.
The matrix product can also be expressed using basic operators of relational algebra: $R_1\Join R_2$ denotes the set of tuples $(i,j,k)$ where $(i,j)\in R_1$ and $(j,k)\in R_2$, and the projection operator $\pi$ can be used to compute the tuples $(i,k)$ where there exists a tuple of the form $(i,\cdot,k)$ in $R_1\Join R_2$.
Since most of our applications are in database systems we will primarily use the notation  of relational algebra.

We consider the following question: given relations $R_1$ and $R_2$ with schemas $(a,b)$ and $(b,c)$, estimate the number $z$ of {\em distinct\/} tuples in the relation $Z=\pi_{ac}(R_1\Join R_2)$.
This problem has been referred to in the literature as \emph{join-project} or \emph{join-distinct}\footnote{Readers familiar with the database literature may notice that we consider projections that return a set, i.e., that projection is duplicate eliminating. We also observe that any equi-join followed by a projection can be reduced to the case above, having two variables in each relation and projecting away the single join attribute. Thus, there is no loss of generality in considering this minimal case.}.
We define $n_1=|R_1|$, $n_2=|R_2|$, and $n=n_1+n_2$.
As observed above, the join-project problem is equivalent to the problem of estimating the number of non-zero entries in the product of two boolean matrices, having $n_1$ and $n_2$ non-zero entries, respectively.

In recent years there has been several papers presenting new algorithms for sparse matrix multiplication~\cite{amossenpagh09joinproject,conf/esa/Lingas09,mm_yuster_zwick}. 
In particular, these algorithms can be used to implement boolean matrix multiplication.
However, the proposed algorithms all have substantially superlinear time complexity in the input size $n$: On worst-case inputs they require time $\omega(n^{4/3})$, even when $z=\bigo(n)$.

In an influential work, Cohen~\cite{JCSS::Cohen1997} presented an estimation algorithm that, for any constant error probability $\delta > 0$, and any $\varepsilon > 0$, can compute a $1\pm \varepsilon$ approximation of $z=|Z|$ in time $\bigo(n/\varepsilon^2)$.
Cohen's algorithm applies to the more general problem of computing the size of the transitive closure of a graph. 

Our main result is that in the special case of sparse matrix product size estimation, we can improve this to expected time $\bigo(n)$ for $\varepsilon > 4/\sqrt[4]{z}$. This means that we have a linear time algorithm for relative error where Cohen's algorithm would use time $\bigo(n\sqrt{z})$.

\paragraph{Approach.}
To build intuition on the size estimation question, consider the sets 
$\mcA_j=\{ i \; | (i,j)\in R_1\}$ and $\mcC_j=\{ k \; | (j,k)\in R_2\}$.
By definition, $Z = \bigcup_j \mcA_j \times \mcC_j$. 
The size of $Z$ depends crucially on the extent of overlap among the sets $\{ \mcA_j \times \mcC_j \}_j$. 
However, the total size of these sets may be much larger than both input and output (see~\cite{amossenpagh09joinproject}), so any approach that explicitly processes them is unattractive.

The starting point for our improved estimation algorithm is a well-known algorithm for estimating the number of distinct elements in a data streaming context~\cite{baryossef02counting}.
(We remark that the idea underlying this algorithm is similar to that of Cohen~\cite{JCSS::Cohen1997}.) Our main insight is that this algorithm can be extended such that a set of the form $\mcA_j \times \mcC_j$ can be added to the sketch in expected time $\bigo(|\mcA_j|+|\mcC_j|)$, i.e., without explicitly generating all pairs.
The idea is to use a hash function that is particularly well suited for the purpose: sufficiently structured to make hash values easy to handle algorithmically, and sufficiently random to make the analysis of sketching accuracy go through.

\subsection{Motivation}

Cohen~\cite{journals/jco/Cohen98} investigated the use of the size estimation technique in sparse matrix computations.
In particular, it can be used to find the optimal order of multiplying sparse matrices, and in memory allocation for sparse matrix computations.

In addition, we are motivated by applications in database systems, where size estimation is an important part of query optimization.
Examples of database queries that correspond to boolean matrix products are:
\begin{itemize}
\item A query that computes all pairs of people in a social network with a distance 2 connection (``possible friends'').
\item A query to compute all director-actor pairs who have done at least one movie together.
\item In a business database with information on orders, and a categorization of products into types, compute the relation that contains a tuple $(c,p)$ if customer $c$ has made an order for a product of type $p$.
\end{itemize}

As a final example, we consider a fundamental data mining task.
Given a list of sets, the famous Apriori data mining algorithm~\cite{apriori} finds frequent item pairs by counting the number occurrences of item pairs where each single element is frequent.
So if $R_1=R_2$ denotes the relationship between high-support (i.e., frequent) items and sets in which they occur, $Z$ is exactly the pairs of frequent items, and the number of distinct items in $Z$ determines the space usage of Apriori.
Since Apriori may be very time consuming, it is of interest to establish whether sufficient space is available before choosing the support threshold and running the algorithm.


\medskip

\subsection{Further related work}

\subsubsection{JD sketch.} %


Ganguly et al.~\cite{ganguly05aggregateestimation} previously considered techniques that compute a data structure (a \emph{sketch}) for $R_1$ and $R_2$ (individually), such that the two sketches suffice to compute an approximation of $z$. 

Define $n_a = |\{ i \; | \; \exists j . (i,j)\in R_1 \}|$ and $n_c = |\{ k \; | \; \exists j . (j,k)\in R_2 \}|$.
Ganguly et al.~show that for any constant $c$ and any $\beta$, a sketching method that returns a $c$-approximation with probability $\Omega(1)$ whenever $z\geq\beta$ must, on a worst-case input, use expected space
$$\Omega(\min(n_1+n_2,n_a n_c (n_1/n_a + n_2/n_c) / \beta)) = \Omega(\min(n_1+n_2,(n_1 n_c + n_2 n_a) / \beta)) \text{ bits.}$$
The lower bound proof applies to the case where $n_1=n_2$, $n_a=n_c$, and $z < n_a + n_c$.
We note that~\cite{ganguly05aggregateestimation} claims a stronger lower bound, but their proof does not establish a lower bound above $n_1+n_2$ bits.
Ganguly et al.~present a sketch whose worst-case space usage matches the lower bound times polylogarithmic factors (while not stated in~\cite{ganguly05aggregateestimation}, the trivial sketch that stores the whole input can be used to nearly match the first term in the minimum).

In Section~\ref{sec:sketch} we analyze a simple sketch, previously considered in other contexts by Gibbons~\cite{Gibbons:2001:DSH} and Ganguly and Saha~\cite{conf/isaac/GangulyS06}.
It similarly matches the above worst-case bound, but the exact space usage is incomparable to that of~\cite{ganguly05aggregateestimation}. 

The focus of~\cite{ganguly05aggregateestimation} is on space usage, and so the time for updating sketches, and for computing the estimate from two sketches, is not discussed in the paper.
Looking at the data structure description we see that the update time grows linearly with the quantity $s_1$, which is $\Omega(n)$ in the worst case.
Also, the sketch uses a number of summary data structures that are accessed in a random fashion, meaning that the worst case number of I/Os is at least $\Omega(n)$ {\em unless\/} the sketch fits internal memory.
By the above lower bound we see that keeping the sketch in internal memory is not feasible in general.
In contrast, the sketch we consider allows collection and combination of sketches to be done efficiently in linear time and I/O.

\subsubsection{Distinct elements and distinct paths estimation.} %
Our work is related in terms of techniques to papers on estimating the number of distinct items in a data stream (see~\cite{baryossef02counting} and its references).
However, our basic estimation algorithm does not work in a general streaming model, since it crucially needs the ability to access all tuples with a particular value on the join attribute together.

Ganguly and Saha~\cite{conf/isaac/GangulyS06} consider the problem of estimating the number of distinct vertex pairs connected by a length-2 path in a graph whose edges are given as a data stream of $n$ edges.
This corresponds to size estimation for the special case of {\em squaring\/} a matrix (or self-join in database terminology).
It is shown that space $\sqrt{n}$ is required, and that space roughly $\bigo(n^{3/4})$ suffices for constant $\varepsilon$ (unless there are close to $n$ connected components).
The estimation itself is a join-distinct size estimation of a sample of the input having size no smaller than $\bigo(n^{3/4}/\varepsilon^2)$.
Using Cohen's estimation algorithm this would require time $\bigo(n^{3/4}/\varepsilon^4)$, so this is $\bigo(n)$ time only for $\varepsilon > 1/\sqrt[16]{n}$. 

\subsubsection{Join synopses.} %
Acharya et al.~\cite{Acharya:1999:JSA} proposed so-called {\em join synopses\/} that provide a uniform sample of the result of a join.
While this can be used to estimate result sizes of a variety of operations, it does not seem to yield efficient estimates of join-project sizes. 
The reason is that a standard uniform sample is known to be inefficient for estimating the number of distinct values~\cite{PODS00*268}.
In addition, Acharya et al.~assume the presence of a foreign-key relationship, i.e., that each tuple has at most one matching tuple in the other table(s), which is also known as a \emph{snow flake} schema. 
Our method has no such restriction.

\subsubsection{Distinct sampling.} %
Gibbons~\cite{Gibbons:2001:DSH} considered different samples that can be extracted by a scan over the input, and proposed {\em distinct samples}, which offer much better guarantees with respect to estimating the number of distinct values in query results.
Gibbons shows that this technique applies to single relations, and to foreign key joins where the join result has the same number of tuples as one of the relations.
In Section~\ref{sec:sketch} we show that the distinct samples, with suitable settings of parameters, can often be used in our setting to get an accurate estimate of $z=|Z|$.
The processing of a pair of samples to produce the estimate consists of running the efficient estimation algorithm of Section~\ref{sec:algo} on the samples, meaning that this is time- and I/O-efficient.


\section{Our algorithm}\label{sec:algo}

The task is to estimate the size $z$ of $Z=\pi_{ac}(R_1\Join R_2)$.
We may assume that attribute values are $\bigo(\log n)$-bits integers, since any domain can be mapped into this one using hashing, without changing the join result size with high probability.
When discussing I/O bounds, $B$ is the number of such integers that fits in a disk block.
In linear expected time (by hashing) or sort$(n)$ I/Os we can cluster the relations according to the value of the join attribute $b$.
By initially eliminating input tuples that do not have any matching tuples in the other relation we may assume without loss of generality that $z\geq n/2$.

In what follows, $k$ is a positive integer parameter that determines the space usage and accuracy of our method.
The technique used is to compute the $k$th smallest value $v$ of a hash function $h(x,y)$, for $(x,y)\in Z$.
Analogously to the result by Bar-Yossef et al.~\cite{baryossef02counting} we can then use $\tilde{z} = k/v$ as an estimator for $z$.

Our main building block is an efficient iteration over all tuples $(x,\cdot,y)\in R_1\Join R_2$ for which $h(x,y)$ is smaller than a carefully chosen threshold $p$, and is therefore a candidate for being among the $k$ smallest hash values.
The essence of our result lies in how the pairs being output by this iteration are computed in expected linear time.
We also introduce a new buffering trick to update the sketch in expected amortized $\bigo(1)$ time per pair.
In a nutshell, each time $k$ new elements have been retrieved, they are merged using a linear time selection procedure with the previous $k$ smallest values to produce a new (unordered) list of the $k$ smallest values.


\begin{theorem}\label{thm:main}
Let $R_1(a,b)$ and $R_2(b,c)$ be relations 
with $n$ tuples in total, and define $z=|\pi_{ac}(R_1\Join R_2)|$.
Let $\varepsilon$, $0 < \varepsilon < \tfrac{1}{4}$ be given.
There are algorithms that run in expected $\bigo(n)$ time on a RAM, and expected $\bigo(\text{sort}(n))$ I/Os in the cache-oblivious model, and output a number $\tilde{z}$ such that for $k = 9/\varepsilon^2$:
\begin{itemize}
	\item $\Pr[(1-\varepsilon)z < \tilde{z} < (1+\varepsilon)z] \geq 2/3$ when $z>k^2$, and
	\item $\Pr[\tilde{z} < (1+\varepsilon) k^2] \geq 2/3$ when $z\leq k^2$.
\end{itemize}
\end{theorem}
Observe that for $\varepsilon > 4/\sqrt[4]{z}$ we will be in the first case, and get the desired $1\pm\varepsilon$ approximation with probability $2/3$.
The error probability can be reduced from $1/3$ to $\delta$ by the standard technique of doing $\bigo(\log(1/\delta))$ runs and taking the median (the analysis follows from a Chernoff bound).
We remark that this can be done in such a way that the $\bigo(\log(1/\delta))$ factor affects only the RAM running time and not the number of I/Os. For constant relative error $\varepsilon > 0$ we have the following result:
\begin{theorem}\label{thm:const-error}
In the setting of Theorem~\ref{thm:main}, if $\varepsilon$ is constant there are algorithms that run in expected $\bigo(n)$ time on a RAM, and expected $\bigo(\text{sort}(n))$ I/Os in the cache-oblivious model, that output $\tilde{z}$ such that $\Pr[(1-\varepsilon)z < \tilde{z} < (1+\varepsilon)z] = 1 - \bigo(1/\sqrt{n})$.
\end{theorem}
The error probability can be reduced to $n^{-c}$ for any desired constant $c$ by running the algorithms $\bigo(c)$ times, and taking the median as above.

\subsection{Finding pairs}
For $\mcB=\pi_b(R_1) \cup \pi_b(R_2)$ and each $i\in \mcB$ let $\mcA_i = \pi_a(\sigma_{b=i}(R_1))$ and $\mcC_i = \pi_c(\sigma_{b=i}(R_2))$.
We would like to efficiently iterate over all pairs $(x,y)\in \mcA_i \times \mcC_i$, $i\in \mcB$, for which $h(x,y)$ is smaller than a threshold $p$.
This is done as follows (see Al\-go\-rithm~\ref{code:pair} for pseudocode).

For a set $U$, let $h_1, h_2:U\rightarrow [0;1]$ be hash functions chosen independently at random from a pairwise independent family, and define $h:U\times U \rightarrow [0;1]$ by\footnote{We observe that this is different from the ``composable hash functions'' used by Ganguly et al.~\cite{ganguly05aggregateestimation}.}
$$h(x,y) = (h_1(x) - h_2(y)) \mod 1.$$
It is easy to show that $h$ is also a pairwise independent hash function --- a property we will utilize later.
Now, conceptually arrange the values of $h(x,y)$ in an $|\mcA_i|\times |\mcC_i|$ matrix, and order the rows by increasing values of $h_1(x)$, and the columns by increasing values of $h_2(y)$.
Then the values of $h(x,y)$ will decrease (modulo~1) from left to right, and increase (modulo~1) from top to bottom.

For each $i\in \mcB$, we traverse the corresponding $|\mcA_i|\times |\mcC_i|$ matrix by visiting the columns from left to right, and in each column $t$ finding the row $\bar{s}$ with the smallest value of $h(x_{\bar{s}},y_t)$.
Values smaller than $p$ in that column will be found in rows subsequent to $\bar{s}$.
When all such values have been output, the search proceeds in column $t+1$.
Notice, that if $h(x_{\bar{s}}, y_t)$ was the minimum value in column $t$, then the minimum value in column $t+1$ is found by increasing $\bar{s}$ until $h(x_{\bar{s}},y_{t+1}) < h(x_{(\bar{s}-1)\mod |\mcA_i|},y_{t+1})$.
We observe that the algorithm is robust to decreasing the value of the threshold $p$ during execution, in the sense that the algorithm still outputs all pairs with hash value at most $p$.
\begin{algorithm*}[!h]
\caption{ Pseudocode for the size estimator.}\label{code:pair}
\begin{algorithmic}[1]
\Procedure{DisItems}{$p,\varepsilon$}
  \State $k \gets \lceil 9/\varepsilon^2 \rceil$
  \State $F\gets \emptyset$
  \For{$i\in \mcB$}
    \State $x \leftarrow \mcA_i$ sorted according to $h_1$-value
    \State $y \leftarrow \mcC_i$ sorted according to $h_2$-value
    \State $\bar{s}\leftarrow 1$
    \For{$t:=1$ to $|\mcC_i|$}
      \While{$h(x_{\bar{s}},y_t) > h(x_{(\bar{s}-1)\mod |\mcA_i|},y_t)$}\Comment{Find $\bar{s}$ s.t.~$h(x_{\bar{s}},y_t)$ is min.}\label{sbar-loop-start}
        \State $\bar{s}\leftarrow (\bar{s}+1)\mod |\mcA_i|$
      \EndWhile\label{sbar-loop-end}
      \State $s\leftarrow \bar{s}$
      \While{$h(x_s,y_t) < p$}\Comment{Find all $s$ where $h(x_s,y_t) < p$}\label{inner-loop-start}
        \State $F \gets F\cup \{(x_{s},y_t)\}$\label{buffer-start}
        \If{$|F| = k$}\Comment{Buffer filled, find smallest hash values in $S\cup F$}
           \State $(p,S) \gets \textrm{\textsc{Combine}}(S,F)$
           \State $F \gets \emptyset$
        \EndIf\label{buffer-end}
        \State $s\leftarrow (s+1)\mod |\mcA_i|$
      \EndWhile\label{inner-loop-end}
    \EndFor
  \EndFor
  \State $(p,S) \gets \textrm{\textsc{Combine}}(S,F)$
  \If{$|S| = k$}
    \State \Return ``$\tilde{z}=\frac{k}{p}$ and $\tilde{z} \in \left[( 1\pm \varepsilon)z \right]$ with probability 2/3'' \label{exact-return}
  \Else 
    \State \Return ``$\tilde{z}=k^2,\;z \leq k^2$ with probability 2/3''
  \EndIf
\EndProcedure
\medskip
\Procedure{Combine}{$S,F$}
  \State $v \gets \textrm{\textsc{Rank}} (h(S) \cup h(F),k)$ \Comment{$\textsc{Rank}(\cdot,k)$ returns the $k$th smallest value}
  \State $S \gets \{x \in S \cup F | h(x) \leq v\}$
  \State \Return $(v,S)$
\EndProcedure
\end{algorithmic}
\end{algorithm*}

\subsection{Estimating the size}
While finding the relevant pairs, we will use a technique that allows us to maintain the $k$ smallest hash values in an unordered buffer instead of using a heap data structure (lines \ref{buffer-start}--\ref{buffer-end} in Algorithm~\ref{code:pair}).
In this way we are able to maintain the $k$ smallest hash values in constant amortized time per insertion in the buffer, eliminating the $\log k$ factor implied by the heap data structure.

Let $S$ and $F$ be two unordered sets containing, respectively, the $k$ smallest hash values seen so far (all, of course, smaller than $p$), and the latest up to $k$ elements seen.
We avoid duplicates in $S$ and $F$ (i.e., the sets are kept disjoint) by using a simple hash table to check for membership before insertion.
Whenever $|F|=k$ the two sets $S$ and $F$ are combined in order to obtain a new sketch $S$.
This is done by finding the median of $S \cup F$, which takes $\bigo(k)$ time using either deterministic methods (see~\cite{zwick_med}) or more practical randomized ones~\cite{Hoare_med}.

At each iteration the current $k$th smallest value in $S$ may be smaller than the initial value $p$, and we use this as a better substitute for the initial value of~$p$.
However, in the analysis below we will upper bound both the running time and the error probability using the initial threshold value $p$.

\subsection{Time analysis}

We split the time analysis into two parts.
One part accounts for iterations of the inner while loop in lines \ref{inner-loop-start}--\ref{inner-loop-end}, and the other part accounts for everything else.
We first consider the RAM model, and then outline the analysis in the cache-oblivious model.

\paragraph{Inner while loop.}

Observe that for each iteration, one pair $(x_s,y_t)$ is added to~$F$ (if it is not already there).
For each $t\in \mcC_i$, $p |\mcA_i|$ elements are expected to be added since each pair $(x_s,y_t)$ is added with probability $p$.
This means that the expected total number of iterations is $\bigo(p |\mcA_i| |\mcC_i|)$.
Each call to {\sc Combine} costs time $\bigo(k)$, but we notice that there must be at least $k$ iterations between successive calls, since the size of $F$ must go from $0$ to $k$.
Inserting a new value into $F$ costs $\bigo(1)$ since the set is not sorted.
Hence, the total cost of the inner loop is $\bigo(p |\mcA_i| |\mcC_i|)$.

\paragraph{Remaining cost.}

Consider the processing of a single $i\in \mcB$ in Algorithm~\ref{code:pair}.
The initial sorting of hash values can be done with bucket sort requiring expected time $\bigo(|\mcA_i| + |\mcC_i|)$ since the numbers sorted are pairwise independent (by the same analysis as for hashing with chaining). 

For the iteration in lines \ref{sbar-loop-start}--\ref{sbar-loop-end} observe that $h(x_{\bar{s}},y_t)$ is monotone modulo 1, and we have at most a total of $2|\mcA_i|$ increments of $\bar{s}$ among all $t\in \mcC_i$.
Thus, the total number of iterations is $\bigo(|\mcA_i|)$, and the total cost for each $i\in \mcB$ is $\bigo(|\mcA_i|+|\mcC_i|)$.

The time for the final call to {\sc Combine} is dominated by the preceding cost of constructing $S$ and $F$.

\paragraph{I/O efficient variant.}

As for I/O efficiency, notice that a direct implementation of Algorithm~\ref{code:pair} may cause a linear number of cache misses if $\mcA_i$ and $\mcC_i$ do not fit into internal memory.
To get an I/O-efficient variant we use a cache-oblivious sorting algorithm, sorting $R_1$ according to $(b,h_1(a))$, and $R_2$ according to $(b,h_2(c))$, such that the sorting steps for each $i\in \mcB$ is replaced by one global sorting step.

The rest of the algorithm works directly in a cache-oblivious setting.
To see this, notice that it suffices to keep in internal memory the two input blocks that are closest to each of the pointers $s$, $t$, and $\bar{s}$.
The cache-oblivious model assumes the cache to behave in an optimal fashion, so also in this model there will be $\Omega(B)$ operations between cache misses, and $\bigo(n/B)$ I/Os, expected, in total.

\begin{lemma}\label{lem:time}
Suppose $R_1(a,b)$ and $R_2(b,c)$ are relations with $n$ tuples in total.
Let $p>0$ and $\varepsilon > 0$ be given.
Then Algorithm~\ref{code:pair} runs in expected $\bigo(n+\sum_i p |\mcA_i| |\mcC_i|)$ time and $\bigo(1/\varepsilon^2)$ space on a RAM, and can be modified to use expected $\bigo(\text{sort}(n))$ I/Os in the cache-oblivious model.
\end{lemma}

\subsubsection{Choice of threshold $p$.}

We would like a value of $p$ that ensures the expected processing time is $\bigo(n)$.
At the same time $p$ should be large enough that we expect to reach line~\ref{exact-return} where an exact estimate is returned (except possibly in the case where $z$ is small).

\begin{lemma}\label{lemma:p_acmax}
Let $j\in \mcB$ satisfy $|\mcA_i||\mcC_i| \leq |\mcA_j||\mcC_j|$ for all $i\in \mcB$.
Then $p=\min(1/k, k/(|\mcA_j||\mcC_j|))$ gives an expected $\bigo(n)$ running time for Algorithm~\ref{code:pair}.
\end{lemma}
\begin{proof}
We argue that for each $i$, $p |\mcA_i| |\mcC_i| \leq \max(|\mcA_i|,|\mcC_i|)$, which by Lemma~\ref{lem:time} implies running time $\bigo(n+\sum_i p |\mcA_i| |\mcC_i|) = \bigo(n+\sum_i \max(|\mcA_i|,|\mcC_i|)) = \bigo(n)$.
Suppose first that $|\mcA_i||\mcC_i| \geq k^2$. Then $p = k/(|\mcA_j||\mcC_j|)$ and $p |\mcA_i| |\mcC_i| \leq k \leq \sqrt{|\mcA_i| |\mcC_i|} \leq \max(|\mcA_i|,|\mcC_i|)$.
Otherwise, when $|\mcA_i||\mcC_i| < k^2$, we have $p \leq 1/k$ and $p |\mcA_i| |\mcC_i| = |\mcA_i| |\mcC_i| / k \leq \max(|\mcA_i|,|\mcC_i|)$. \qed
\end{proof}

We note that when $R_1$ and $R_2$ are sorted according to $b$, the value of $p$ specified above can be found by a simple scan over both inputs.
Our experiments indicate that in practice this initial scan is not needed, see Section~\ref{sec:experiments} for details.

\subsection{Error probability}\label{sec:bar-yossef}
\begin{theorem}\label{theorem:kmin}
Let $h$ be a pairwise independent hash function.
Suppose we are provided with a stream of elements $N$ with $h(x) < v$ for all $x\in N$.
Further, let~$\varepsilon$, $0< \varepsilon < \tfrac{1}{4}$ be given and assume that $p \geq \min \left( \frac{k}{2z}, \frac{1}{k} \right)$, where $k\geq 9/\varepsilon^2$, and $z$ is the number of distinct items in $N$.
Then Algorithm~\ref{code:pair} produces an approximation $\tilde{z}$ of $z$ such that
\begin{itemize}
	\item $\Pr[(1-\varepsilon)z < \tilde{z} < (1+\varepsilon)z] \geq 2/3$ for $z>k^2$, and
	\item $\Pr[\tilde{z} < (1+\varepsilon) k^2] \geq 2/3$ for $z\leq k^2$.
\end{itemize}
\end{theorem}
\begin{proof}
The error probability proof is similar to the one that can be found in~\cite{baryossef02counting}, with some differences and extensions.
We bound the error probability of three cases: the estimate being smaller/larger than the multiplicative error bound, and the number of obtained samples being too small.

\emph{Estimate too large}.\quad Let us first consider the case where $\tilde{z} > (1+\varepsilon)z$, i.e.~the algorithm overestimates the number of distinct elements.
This happens if the stream $N$ contains at least $k$ entries smaller than $k/(1+\varepsilon)z$.
For each pair $(a,c)\in Z$ define an indicator random variable $\Xac$ as
$$
\Xac = \begin{cases}1 & h(a,c) < k/(1+\varepsilon)z\\
     0 & \text{otherwise}
\end{cases}
$$
That is, we have $z$ such random variables for which the probability of $\Xac = 1$ is exactly $k/(1+\varepsilon)z$
\ and $\E[\Xac] = k/(1+\varepsilon)z$.
Now define $Y=\sum_{(a,c)\in Z}\Xac$ so that $\E[Y] = \E[\sum_{(a,c)\in Z}\Xac] = \sum_{(a,c)\in Z} \E[\Xac]=k/(1+\varepsilon)$.
By the pairwise independence of the $\Xac$ we also get $\Var(Y) \leq k/(1+\varepsilon)$.
Using Chebyshev's inequality~\cite{books/cu/MotwaniR95} we can bound the probability of having too many pairs reported:
$$
 \Pr\left[Y > k\right] \leq \Pr\left[|Y-\E[Y]| > k-\tfrac{k}{1+\varepsilon}\right]
  \leq \frac{\Var[Y]}{\left(k-\tfrac{k}{1+\varepsilon}\right)^2}
   \leq \frac{k/(1+\varepsilon)}{\left(k-\tfrac{k}{1+\varepsilon}\right)^2} \leq \tfrac{1}{6}
$$
since $k\geq 9/\varepsilon^2$.

\emph{Estimate too small}.\quad Now, consider the case where $\tilde{z} < (1-\varepsilon)z$ which happens when at most $k$ hash values are smaller than $k/(1-\varepsilon) z$ and at least $k$ hash values are smaller than $p$.
Define $\Xac'$ as 
$$
\Xac' = \begin{cases}1 & h(a,c) < k/(1-\varepsilon)z\\
     0 & \text{otherwise}
\end{cases}
$$
so that $\E[\Xac'] = k/(1-\varepsilon) z < (1+\varepsilon) k /z$.
Moreover, with $Y'=\sum_{(a,c)\in Z} \Xac'$ we have $\E[Y'] = k/(1-\varepsilon)$, and since the indicator random variables defined above are pairwise independent, we also have $\Var[Y'] \leq \E[Y'] < (1+\varepsilon) k$.
Chebyshev's inequality gives:
$$
 \Pr\left[Y' < k\right] \leq \Pr\left[|Y'-\E[Y']| > \tfrac{k}{1-\varepsilon}-k\right]
 \leq \frac{\Var[Y']}{\left(k-\tfrac{k}{1+\varepsilon}\right)^2}
 \leq \frac{(1+\varepsilon)k}{\left(\tfrac{k}{1-\varepsilon}-k\right)^2} < \tfrac{1}{9}
$$
since $k\geq 9/\varepsilon^2$.

\emph{Not enough samples.}\quad Consider the case where $|S|<k$ after all pairs have been retrieved.
In this case the algorithm returns $\beta = k^2$ as an upper bound on the number of distinct elements in the output, and we have two possible situations:
either there is actually less than $k^2$ distinct pairs in the output, in which case the algorithm is correct, or there are more than $k^2$ distinct elements in the output, in which case it is incorrect.
In the latter case, less than $k$ hash values have been smaller than $p$ and the $k$th smallest value $v$ is therefore larger than $p$.
Define $\Xac''$ as 
$$
\Xac'' = \begin{cases}1 & h(a,c) < p\\
     0 & \text{otherwise}
\end{cases}
$$
and let again $Y''=\sum_{(a,c)\in Z} \Xac''$.
It results that $\E[\Xac'']=p$ and $\E[Y''] = zp$, and because of pairwise independancy of $\Xac''$, also $\Var[Y''] \leq \E[Y'']$.
Using Chebyshev's inequality and remembering that $z>k^2$ in this case we have:
$$
\Pr[Y'' < k] \leq \Pr[|Y''-\E[Y'']| > zp-k] 
 \leq \frac{zp}{(zp-k)^2} \leq \frac{zp}{\left( \frac{1}{2}zp \right)^2} \leq 8/k \leq 1/18.
$$
using that $k \geq 9/\varepsilon^2 \geq 144$.

In conclusion, the probability that the algorithm fails to output an estimate within the given limits is at most $1/6 + 1/9 + 1/18 = 1/3$. \qed
\end{proof}

For the proof of Theorem~\ref{thm:const-error} we observe that in the above proof, if $\varepsilon$ is constant the error probability is $\bigo(1/k)$. Using $k=\sqrt{n}$ we get linear running time and error probability $\bigo(1/\sqrt{n})$.

\subsubsection{Realization of hash functions.}
We have used the idealized assumption that hash values were real numbers in $(0;1)$.
Let $m=n^3$.
To get an actual implementation we approximate (by rounding down) the real numbers used by rational numbers of the form $i/m$, for integer $i$.
This changes each hash value by at most $2/m$.
Now, because of the way hash values are computed, the probability that we get a different result when comparing two real-valued hash values and two rational ones is bounded by $2/m$.
Similarly, the probability that we get a different result when looking up a hash value in the dictionary is bounded by $2k/m$.
Thus, the probability that the algorithm makes a different decision based on the approximation, in any of its steps, is $\bigo(kn/m) = o(1)$.
Also, for the final output the error introduced by rounding is negligible.

\section{Distinct sketches}\label{sec:sketch}

A well-known approach to size estimation in, described in generality by Gibbons~\cite{Gibbons:2001:DSH} and explicitly for join-project operations in~\cite{conf/isaac/GangulyS06,amossenpagh09joinproject}, is to sample random subsets $R'_1\subseteq R_1$ and $R'_2\subseteq R_2$, compute $Z'=\pi_{ac}(R'_1\Join R'_2)$, and use the size of $Z'$ to derive an estimate for $z$.
This is possible if $R'_1=\sigma_{a\in S_a}(R_1)$, where $S_a \subseteq \pi_a(R_1)$ is a random subset where each element is picked independently with probability $p_1$, and similarly $R'_2=\sigma_{c\in S_c}(R_2)$, where $S_c\subseteq \pi_c(R_2)$ includes each element independently with probability $p_2$.
Then $z'=|Z'|/(p_1 p_2)$ is an unbiased estimator for $z$.
The samples can be obtained in small space using hash functions whose values determine which elements are picked for $S_a$ and $S_c$.
The value $|Z'|$ can be approximated in linear time using the method described in section~\ref{sec:algo} if the samples are sorted --- otherwise one has to add the cost of sorting.
In either case, the estimation algorithm is I/O-efficient.

Below we analyze the variance of the estimator $z'$, to identify the minimum sampling probability that introduces only a small relative error with good probability.
The usual technique of repetition can be used to reduce the error probability.
Recall that we have two relations with $n_1$ and $n_2$ tuples, respectively, and that $n_a$ and $n_c$ denotes the number of distinct values of attributes $a$ and $c$, respectively.
Our method will pick samples $R'_1$ and $R'_2$ of expected size $s$ from each relation, where $s=p_1 n_1=p_2 n_2$ is a parameter to be specified.

\begin{theorem}\label{thm:sketch}
Let $R'_1$ and $R'_2$ be samples of size $s$, obtained as described above.
Then $z'=|\pi_{ac}(R'_1\Join R'_2)|/(p_1 p_2)$ is a $1\pm\varepsilon$ approximation of $z=|\pi_{ac}(R_1\Join R_2)|$ with probability $5/6$ if $z>\beta$, where $\beta = \frac{14}{\varepsilon^2}\left(\frac{n_c n_1 + n_a n_2}{s} \right)$.
If $z\leq\beta$ then $z'<(1+\varepsilon)\beta$ with probability $5/6$.
\end{theorem}

\subsection{Analysis of variance}

To arrive at a sufficient condition that $z'$ is a $1\pm\varepsilon$ approximation of $z$ with good probability, we analyze its variance. 
To this end define $Z_{i\cdot } = \{ j \, | \, (i,j)\in Z\}$, $Z_{\cdot j} = \{ i \, | \, (i,j)\in Z\}$, and let
$$
 X_i = \left\{ \begin{array}{ll}
 1-p_1, & \text{if }i\in S_a\\
 -p_1, & \text{otherwise}
 \end{array} \right.
\qquad
 Y_j = \left\{ \begin{array}{ll}
 1-p_2, & \text{if }j\in S_c\\
 -p_2, & \text{otherwise}
 \end{array} \right. .
$$
By definition of $S_a$, $\E[X_i] = \Pr[i\in S_a] (1-p_1) - \Pr[i\not\in S_a] p_1 = 0$. 
Similarly, $\E[Y_i]=0$. 
We have that $(i,j)\in Z'$ if and only if $(i,j)\in Z$ and $(i,j)\in S_a\times S_c$. 
This means that $z' p_1 p_2 = \sum_{(i,j)\in Z} (X_i + p_1)(Y_j + p_2)$.
By linearity of expectation, $\E[(X_i + p_1)(Y_j + p_2)] = p_1 p_2$, and we can write the variance of $z' p_1 p_2$, $\Var(z' p_1 p_2)$ as
\begin{displaymath}
\E\left[ \left(\sum_{(i,j)\in Z} \left((X_i + p_1)(Y_j + p_2) - p_1 p_2\right) \right)^2 \right].
\end{displaymath}
Expanding the product and using linearity of expectation, we get
\begin{alignat*}{1}
\Var(z' p_1 p_2) & = \sum_{(i,j)\in Z} \sum_{(i,j')\in Z} \E\left[ X_i^2 p_2^2 \right] 
    + \sum_{(i,j)\in Z} \sum_{(i',j)\in Z} \E\left[ Y_j^2 p_1^2 \right] 
    + \sum_{(i,j)\in Z} \E\left[ X_i^2 Y_j^2 \right]\\
    & = \sum_{i\in \mcA}\sum_{j,j'\in Z_{i\cdot }} p_2^2\, \E\left[ X_i^2 \right]
     + \sum_{j\in \mcC}\sum_{i,i'\in Z_{\cdot j}} p_1^2\, \E\left[ Y_i^2 \right]
     + z\, \E\left[ X_i^2 \right] \E\left[ Y_i^2 \right]
\end{alignat*}
Since $\E\left[ X_i^2 \right] = p_1 (1-p_1)^2  + (1-p_1) (-p_1)^2 = p_1 - p_1^2 < p_1$, and similarly $\E\left[ Y_j^2 \right] < p_2$ we can upper bound $\Var(z')$ as follows:
\begin{alignat*}{1}
\Var(z') & = (p_1 p_2)^{-2} \,\Var(z' p_1 p_2)\\
 			   & < (p_1 p_2)^{-2} 
				  \Big( \sum_{i\in \mcA}\sum_{j,j'\in Z_{i\cdot }} p_1 p_2^2 + \sum_{j\in \mcC}\sum_{i,i'\in Z_{\cdot j}} p_1^2 p_2 + z\, p_1 p_2 \Big)\\
 			   & \leq (p_1 p_2)^{-2} 
				  \left( n_c z\, p_1 p_2^2 + n_a  z\, p_1^2 p_2 + z\, p_1 p_2 \right)\\
			   & = \left(n_c/p_1 + n_a /p_2 + (p_1 p_2)^{-1}\right) \, z \enspace .
\end{alignat*}

\subsection{Sufficient sample size}

We are ready to derive a bound on the probability that~$z'$ deviates significantly from $z$.
Choose $0<\varepsilon < 1$. 
Since $z = \E[z']$ Chebyshev's inequality says
$$
\Pr[|z'-z]>\varepsilon z] < \frac{\Var(z')}{(\varepsilon z)^2} \leq \left(n_c/p_1 + n_a /p_2 + (p_1 p_2)^{-1}\right) / (\varepsilon^2 z).
$$
This can equivalently be expressed in terms of the sample size $s$, since $p_1 = s / n_1$ and $p_2 = s / n_2$:
$$
\Pr[|z'-z]>\varepsilon z] < \left(n_c n_1 + n_a n_2 + n_1 n_2 / s\right) / (s \varepsilon^2 z).
$$
We seek a sufficient condition on $s$ that the above probability is bounded by some  constant $\delta < \tfrac{1}{2}$ (e.g.~$\delta = 1/6$).
In particular it must be the case that 
$n_1 n_2 / (s^2\varepsilon^2 z) < \delta$, which implies $s > \sqrt{n_1,n_2 / (\delta z)} \geq \sqrt{n_1,n_2 / (\delta n_a  n_c)}$.
Hence, using the arithmetic-geometric inequality:
$$
n_1 n_2 / s < \sqrt{n_c n_1 n_a  n_2 \delta} \leq (n_c n_1 + n_a  n_2) / (2\sqrt{\delta}).
$$
In other words, it suffices that
\begin{gather*}
\frac{\left(n_c n_1 + n_a n_2\right) (1+(2\sqrt{\delta})^{-1})}{s \varepsilon^2 z} < \delta
\iff s > \Bigg( \frac{n_c n_1 + n_a n_2}{z} \Bigg) 
\Bigg( \frac{1+(2\sqrt{\delta})^{-1}}{\varepsilon^2 \delta}\Bigg).
\end{gather*}

One apparent problem is the chicken-egg situation: $z$ is not known in advance.
If a lower bound on $z$ is known, this can be used to compute a sufficient sample size.
Alternatively, if we allow a larger relative error whenever $z\leq \beta$ we may compute a sufficient value of $s$ based on the assumption $z \geq \beta$.
Whenever $z < \beta$ we then get the guarantee that $z' < (1+\varepsilon) \beta$ with probability $1-\delta$.
Theorem~\ref{thm:sketch} follows by fixing $s$ and solving for $\beta$.

\subsubsection{Optimality.}
For constant $\varepsilon$ and $\delta$ our upper bound matches the lower bound of Ganguly et al.~\cite{ganguly05aggregateestimation} whenever this does not exceed $n_1+n_2$. It is trivial to achieve a sketch of size $\bigo((n_1+n_2)\log (n_1+n_2))$ bits (simply store hash signatures for the entire relations).
We also note that the lower bound proof in~\cite{ganguly05aggregateestimation} uses certain restrictions of parameters ($n_1=n_2$, $n_a=n_c$, and $z < n_a + n_c$), so it may be possible to do better in some settings.

\section{Experiments}\label{sec:experiments}

\begin{figure}[p]
\centering
\subfigure[$k=256$]{\includegraphics[width=.45\textwidth]{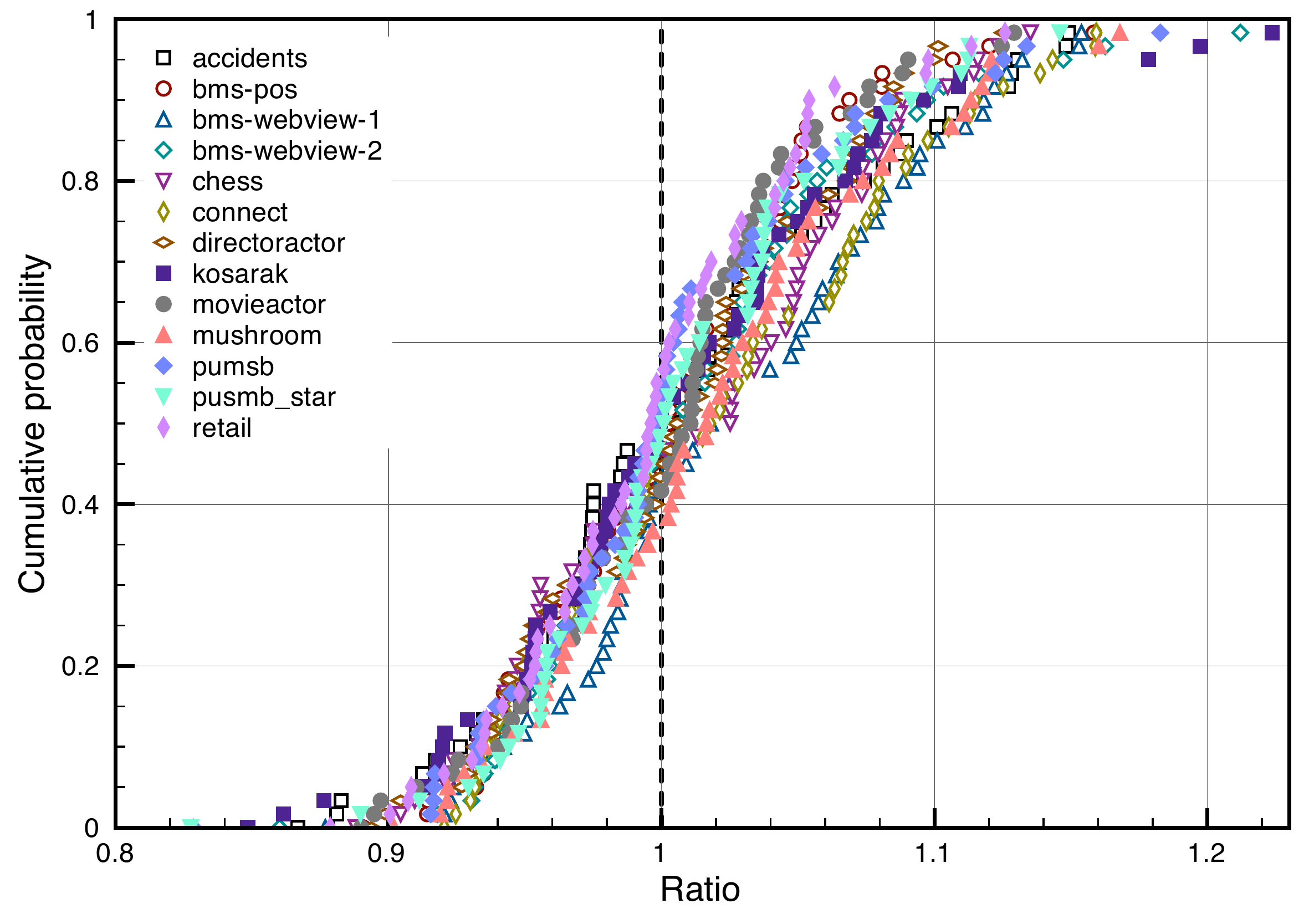}}
\hspace{5mm}
\subfigure[$k=1024$]{\includegraphics[width=.45\textwidth]{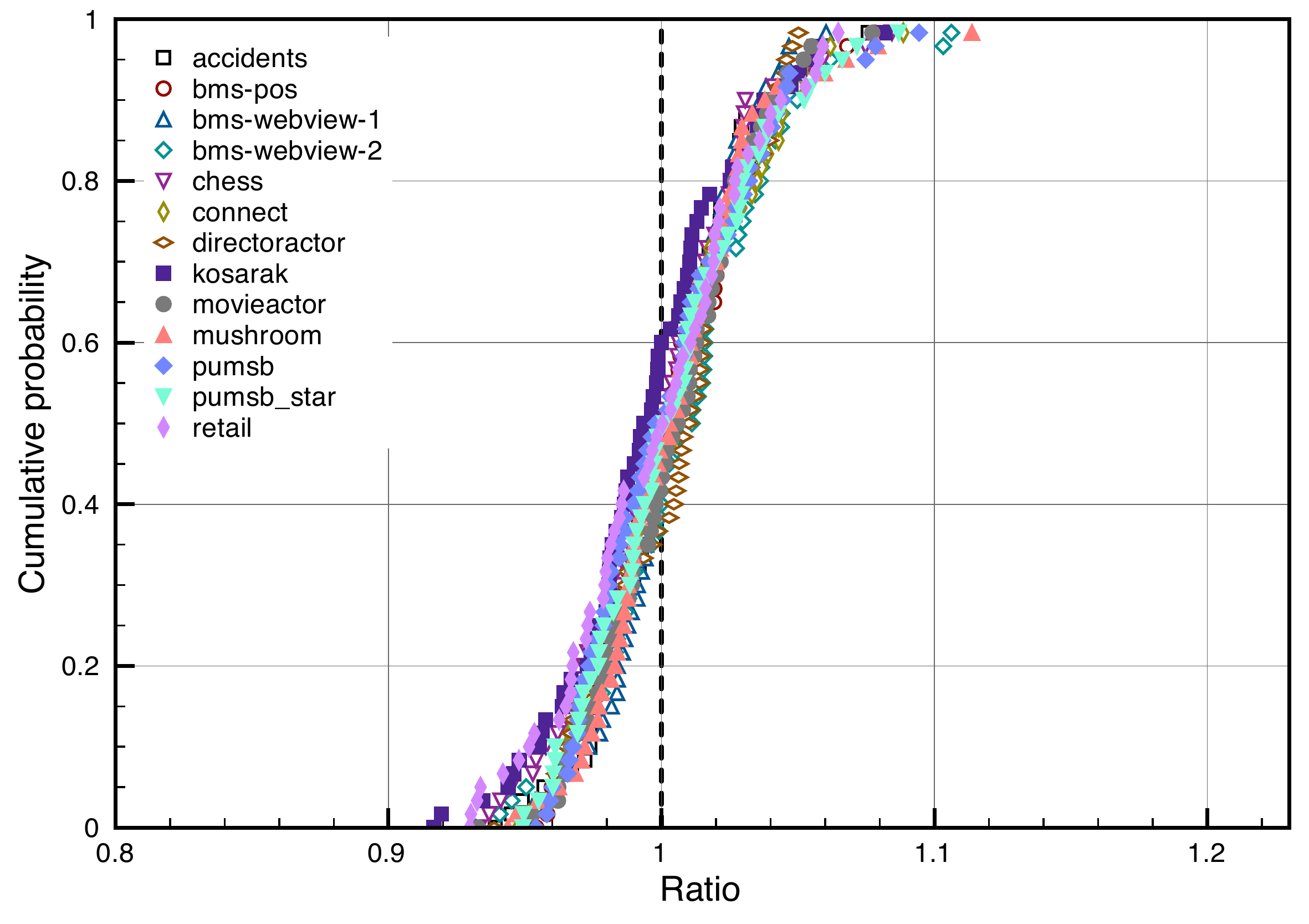}}
\caption{ The cumulative distribution functions for $k=256$ and $k=1024$.
It is seen that $k=1024$ yields a more precise estimate than $k=256$ with 2/3 of the estimates being within $4\%$ and $10\%$ of the exact size, respectively.}\label{fig:k256}
\end{figure}

\begin{figure}[p]
\centering
\subfigure[$k=1024$, $p_1=p_2=0.1$]{\includegraphics[width=.45\textwidth]{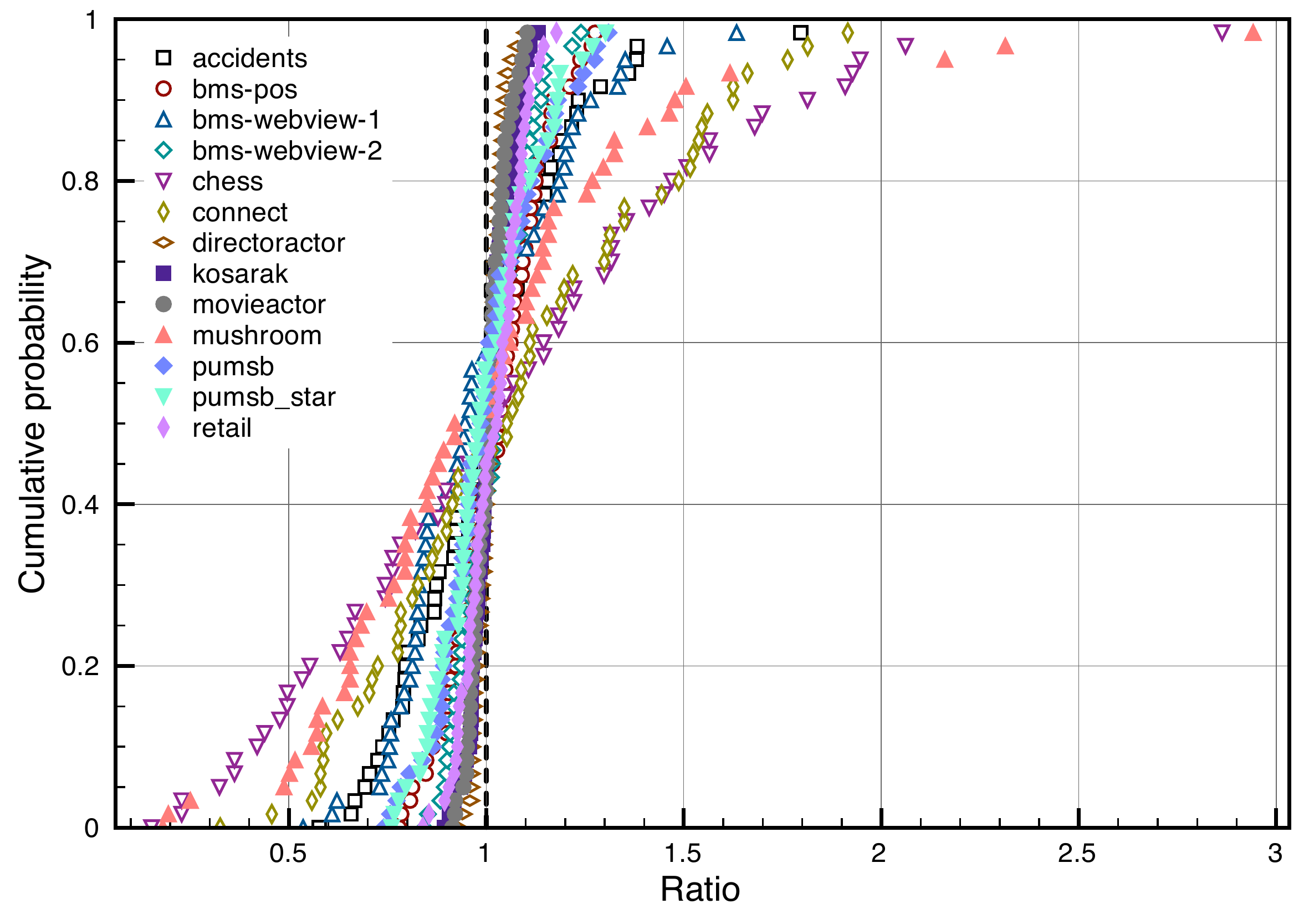}}
\hspace{5mm}
\subfigure[$k=1024$, $p_1=p_2=0.01$]{\includegraphics[width=.45\textwidth]{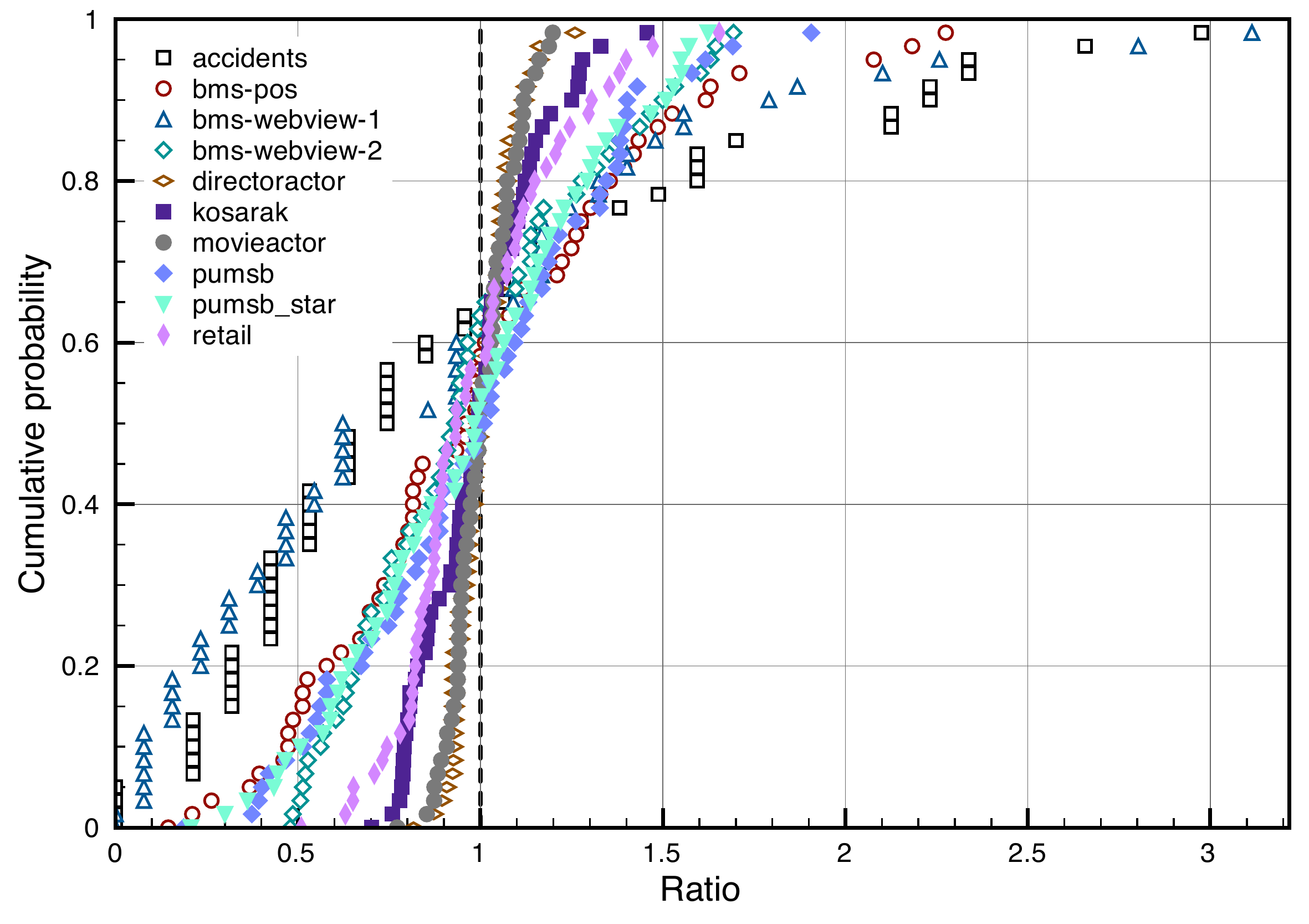}}
\caption{ Plots for sampling with probability 10\% and 1\%.
If the sampling probability is too small, no elements at all may reach the sketch and in these cases we are not able to return an estimate.
Instances with no estimates have been left out of the graph.
}\label{fig:sampling}
\end{figure}

We have run our algorithm on most of the datasets from the Frequent Itemset Mining Implementations (FIMI) Repository\footnote{http://fimi.cs.helsinki.fi} together with some datasets extracted from the Internet Movie Database (IMDB).
Each dataset represents a single relation, and motivated by the Apriori space estimation example in the introduction, we perform the size estimation on self-joins of these relations.
Table~\ref{table:datasets} displays the size of each dataset together with the number of distinct $a$- and $c$-values.

\begin{table}[!h]
\centering
\small
\begin{tabular}{l|cc|cc}
Instance      & $z$             & $n_a\,(= n_c)$  & $\varepsilon_{0.1}$  & $\varepsilon_{0.01}$\\
\hline
accidents     &   $94\cdot10^3$ & 468             & 1.18           & 3.73\\
bms-pos       &  $760\cdot10^3$ & 1,657           & 0.78           & 2.47 \\
bms-webview-1 &  $128\cdot10^3$ & 497             & 1.04           & 3.29\\
bms-webview-2 & $1.45\cdot10^6$ & 3,340           & 0.80           & 2.54\\
chess         & $5.24\cdot10^3$ & 75              & 2.00           & 6.33\\
connect       & $13.8\cdot10^3$ & 129             & 1.62           & 5.12\\
directoractor &  $734\cdot10^6$ & 50,645          & 0.14           & 0.44\\
kosarak       & $66.2\cdot10^6$ & 41,270          & 0.42           & 1.32\\
movieactor    &  $111\cdot10^6$ & 51,226          & 0.36           & 1.14\\
mushroom      & $7.17\cdot10^3$ & 119             & 2.16           & 6.82\\
pumsb         & $1.07\cdot10^6$ & 2,113           & 0.74           & 2.35\\
pumsb\_star   &  $967\cdot10^3$ & 2,088           & 0.78           & 2.46\\
retail        & $7.19\cdot10^6$ & 16,470          & 0.80           & 2.53\\
\end{tabular}
\medskip
\caption{ Characteristics of the used datasets. The rightmost middle column displays the size $n_a = |\pi_{a}(R_1)|$ (which in this case is equals $n_c = |\cup \pi_c(R_2)|$).
The two rightmost columns display the theoretical error as described in Theorem~\ref{thm:sketch}, for $p_1=p_2=0.1$ and $p_1=p_2=0.01$, respectively.
These theoretical error bounds, which hold with probability $5/6$, are significantly larger than the actual observed errors in Figure~\ref{fig:sampling}.
}\label{table:datasets}
\end{table}

Rather than selecting $h_1$ and $h_2$ from an arbitrary pairwise independent family, we store functions that map the attribute values to fully random and independent values of the form $d/2^{64}$, where $d$ is a 64 bit random integer formed by reading 64 random bits from the Marsaglia Random Number CDROM\footnote{http://www.stat.fsu.edu/pub/diehard/}. 

We have chosen an initial value of $p=1$ for our tests in order to be certain to always arrive at an estimate.
In most cases we observed that $p$ quickly decreases to a value below $1/k$ anyway.
But as the sampling probability decreases, the probability that the sketch will never be filled increases, implying that we will not get a linear time complexity with an initial value of $p=1$.
In the cases where the sketch is not filled, we report $|F|/(p_1 p_2)$ as the estimate, where $|F|$ is the number of elements in the buffer.

Tests have been performed for $k = 256$ and $k=1024$.
In each test, 60 independent estimates were made and compared to the exact size of the join-project.
By sorting the ratios ``estimate''/''exact size'' we can draw the cumulative distribution function for each instance that, for each ratio-value on the $x$-axis, displays on the $y$-axis the probability that an estimate will have this ratio or less.
Figure~\ref{fig:k256} shows plots for $k=256$ and $k=1024$.
In Table~\ref{table:error1} we compare the theoretical error $\varepsilon$ with observed error for 2/3 of the results.
As seen, the observed error is smaller than the theoretical upper bound.

In Figure~\ref{fig:sampling} we perform sampling with 10\% and 1\% probability, as described in Section~\ref{sec:sketch}.
Again, the samples are chosen using truly random bits.
The variance of estimates increase as the probability decreases, but increases more for smaller than for larger instances.
If the sampling probability is too small, no elements at all may reach the sketch and in these cases we are not able to return an estimate.
As seen, the observed errors in the figure are significantly smaller than the theoretical errors seen in Table~\ref{table:datasets}.

\begin{table}[!h]
\centering
\begin{tabular}{l|cc}
$k$      & $\epsilon$ & Observed $\varepsilon$\\
\hline
256      & 0.188 & 0.1 \\
1024     & 0.094 & 0.04 \\
\end{tabular}
\medskip
\caption{ The theoretical error bound is $\varepsilon = \sqrt{9/k}$ as stated Theorem~\ref{theorem:kmin}.
The observed error in Figure~\ref{fig:k256}, however, is significantly
less.}\label{table:error1}
\end{table}


\medskip

\section{Conclusion}
We have presented improved algorithms for estimating the size of boolean matrix products, for the first time allowing $o(1)$ relative error to be achieved in linear time. 
An interesting open problem is if this can be extended to transitive closure in general graphs, and/or to products of more than two matrices.

\medskip

{\bf Acknowledgement.} We would like to thank Jelani Nelson for useful discussions, and in particular for introducing us to the idea of buffering to achieve faster data stream algorithms. Also, we thank Sumit Ganguly for clarifying the lower bound proof of~\cite{ganguly05aggregateestimation} to us. Finally, we thank Konstantin Kutzkov and Rolf Fagerberg for pointing out mistakes that have been corrected in this version of the paper.






\bibliographystyle{abbrv} 
\bibliography{bibliography}
\end{document}